\documentclass[12pt]{article}

\usepackage{graphics,graphicx,fullpage,natbib,multirow}
\usepackage{amsmath,amssymb,verbatim,epsfig}
\usepackage[dvipsnames,usenames]{color}
\usepackage{float}
\usepackage[table]{xcolor}
\usepackage{authblk}

\newtheorem{proposition}{Proposition}

\newtheorem{corollary}{Corollary}
\newtheorem{defin}{\bf Definition}

\newenvironment{proof}{\noindent{\bf Proof}}{$\diamond$}

\def\be{\mbox{Be}}
\def\bin{\mbox{Bin}}

\def\un{\mbox{Un}}

\def\bebin{\mbox{BeBin}}

\def\Cor{\mbox{Corr}}

\def\sbe{\mbox{SBeP}}
\def\d{\mbox{d}}
\def\rest{\mbox{rest}}

\def\bc{{\bf c}}

\def\bu{{\bf u}}
\def\bv{{\bf v}}

\def\bX{{\bf X}}
\def\bY{{\bf Y}}

\newcommand{\bfeta}{\boldsymbol{\eta}}

\newcommand{\btheta}{\boldsymbol{\theta}}

\newcommand{\Naa}{{\rm I}\!{\rm N}}

\newcommand{\RR}{\mathbb{R}}

\begin{document}

\baselineskip=24pt

\title{\bf Generalised Bayesian sample copula\\ of order $m$}

\author[$\dagger$]{Luis E. Nieto-Barajas}
\author[$\ddagger$]{Ricardo Hoyos-Arg\"uelles}

\affil[$\dagger$]{\small{Department of Statistics, ITAM, Mexico}}
\affil[$\ddagger$]{Direction of Financial System Information, Banco de M\'exico, Mexico}

\date{}

\maketitle

\begin{abstract}
In this work we propose a semiparametric bivariate copula whose density is defined by a piecewise constant function on disjoint squares. We obtain the maximum likelihood estimators of model parameters and prove that they reduce to the sample copula under specific conditions. We further propose to carry out a full Bayesian analysis of the model and introduce a spatial dependent prior distribution for the model parameters. This prior allows the parameters to borrow strength across neighbouring regions to produce smooth posterior estimates. To characterise the posterior distribution, via the full conditional distributions, we propose a data augmentation technique. A Metropolis-Hastings step is required and we propose a novel adaptation scheme for the random walk proposal distribution. We implement a simulation study and an analysis of a real dataset to illustrate the performance of our model and inference algorithms.
\end{abstract}

\noindent {\sl Keywords}: Copula theory, Bayesian inference, piecewise constant, spatial model.

\noindent {\sl AMS Classification}: 62H05 $\cdot$ 62G05 $\cdot$ 62H11.

\section{Introduction}
\label{sec:intro}

Let $(X,Y)$ be a bivariate random vector with joint cumulative distribution function (CDF) $H(x,y)$ and marginal CDFs $F(x)$ and $G(y)$, respectively. According to \cite{sklar:59}, there exists a copula function $C(u,v)$ with $C:[0,1]^2\to[0,1]$ that satisfies the conditions to be a proper CDF with uniform marginals, such that $H(x,y)=C(F(x),G(y))$.  

Dependence or association measures between the two random variables $(X,Y)$, independently of their marginal distributions, can be entirely written in terms of the copula. For instance Kendall's $\tau$ and Spearman's $\rho$ are given by 
\begin{equation}
\label{eq:tau-rho}
\tau=4\int_0^1\int_0^1 C(u,v)f_C(u,v)\d u \d v-1\quad\mbox{and}\quad\rho=12\int_0^1\int_0^1 uvf_C(u,v)\d u\d v-3,
\end{equation}
respectively, where $f_C(u,v)$ is the corresponding copula density \citep[e.g.][]{nelsen:06}. Therefore, our interest lies on estimating the copula, either the CDF $C(u,v)$ or the density $f_C(u,v)$. 

Nonparametric estimation of copulas was first proposed by \citep{deheuvels:79} who introduced the empirical copula based on the multivariate empirical distribution on the marginal empirical distributions. Later, \cite{fermanian&al:04} studied weak convergence properties of the empirical copula. Smoother estimators were also proposed based on kernels \citep[e.g.][]{fermanian&scaillet:03}. For example \cite{chen&huang:07} proposed a bivariate kernel copula based on local linear kernels that is everywhere consistent on $[0,1]^2$. In a Bayesian perspective, \cite{hoyos&nieto:20} proposed a nonparametric estimator of the generator, in an Archimedean copula, based on quadratic splines. 

Recently, as a generalisation of the empirical copula, \cite{gonzalez&hoyos:18} introduced a sample copula of order $m$ based on a modified rank transformation of the data. Since our model is closely related to the sample copula, we review it here in detail.  
Let $(X_i,Y_i)$, $i=1,\ldots,n$, be a bivariate sample with support $\Omega\subset\RR^2$. Based on the probability integral transformation and using the empirical CDF of each coordinate, the modified rank transformation \citep{deheuvels:79} $(U_i,V_i)$ is defined as 
\begin{equation}
\label{eq:rank}
U_i=\mbox{rank}(i,\bX)/n\quad\mbox{and}\quad V_i=\mbox{rank}(i,\bY)/n,
\end{equation}
where $\mbox{rank}(i,\bX)=k$ if and only if $X_i=X_{(k)}$ for $i,k=1,\ldots,n$. The modified sample $(U_i,V_i)$, $i=1,\ldots,n$ has support in $[0,1]^2$, but contains all relevant information in the data to characterise the copula (dependence). In particular, the original sample $(X_i,Y_i)$ and the modified rank transformed sample $(U_i,V_i)$ produce exactly the same sample Kendall's tau and Spearman's rho coefficients \citep[e.g][]{nelsen:06}. 

Now, independently of the data, let us consider a uniform partition of size $m$, $2\leq m\leq n$, for each of the two coordinates in $[0,1]$. Then $\{Q_{j,k},\:j,k=1,\ldots,m\}$ defines a partition of size $m^2$ of $[0,1]^2$ such that 
\begin{equation}
\label{eq:partition}
Q_{j,k}=\left(\frac{j-1}{m},\frac{j}{m}\right]\times\left(\frac{k-1}{m},\frac{k}{m}\right] 
\end{equation}
is the region in the unitary square formed by the cross product of intervals $j^{th}$ and $k^{th}$ of the first and second coordinate, respectively, for $j,k=1,\ldots,m$. To illustrate, Figure \ref{fig:part} depicts a partition with $m=5$. 
Let $r_{j,k}$ be the number of modified sample points belonging to region $Q_{j,k}$, in notation 
\begin{equation}
\label{eq:r}
r_{j,k}=\sum_{i=1}^n I((u_i,v_i)\in Q_{j,k}),
\end{equation}
for $j,k=1,\ldots,m$ such that $\sum_{j=1}^m\sum_{k=1}^m r_{j,k}=n$. Then the sample copula density of order $m$ is defined as 
\begin{equation}
\label{eq:samplec}
f_S(u,v\mid\btheta)=(m^2/n)\sum_{j=1}^m\sum_{k=1}^m r_{j,k}I\left((u,v)\in Q_{j,k}\right).
\end{equation}
Further properties of this sample copula were studied in \cite{gonzalez&hoyos:21}. 

In this paper we propose a semiparametric copula model whose maximum likelihood estimator coincides, under certain conditions, to the sample copula of order $m$. We further propose a Bayesian approach for inference purposes and introduce a novel prior that borrows strength across some neighbouring regions in the space and produces smooth estimates. Posterior inference is obtained via a Markov Chain Monte Carlo algorithm that relies in a data augmentation technique and requires a Metropolis-Hastings step. We propose a novel adaptation scheme for the random walk proposal distributions.

The outline of the rest of the paper is as follows: In Section \ref{sec:model} we define the semiparametric copula model and obtain the maximum likelihood estimators of the model parameters. In Section 3 we introduce the spatially dependent prior and study its properties. Section 4 characterises posterior distributions and deals with posterior computations. In Section 5 we present a simulation study to show the performance of our model under different scenarios and carry out a real data analysis. We conclude in Section 6 with a discussion.

\section{Model}
\label{sec:model}

Let us consider a uniform partition of size $m^2$, $2\leq m\leq n$, of $[0,1]^2$ as in \eqref{eq:partition}.  We define a semiparametric copula density of the form
\begin{equation}
\label{eq:dens}
f_C(u,v\mid\btheta)=m^2\sum_{j=1}^m\sum_{k=1}^m\theta_{j,k}I\left((u,v)\in Q_{j,k}\right),
\end{equation}
where $\btheta=\{\theta_{j,k},\:j,k=1\ldots,m\}$ are the set of model parameters that satisfy the following conditions: 
\begin{equation}
\label{eq:cond1}
\sum_{j=1}^m\theta_{j,k}=\sum_{k=1}^m\theta_{j,k}=\frac{1}{m}\quad\text{and}\quad\sum_{j=1}^m\sum_{k=1}^m\theta_{j,k}=1.
\end{equation}

Our semiparametric model \eqref{eq:dens} can be seen as a bivariate probability histogram with $m^2$ number of bins, and \eqref{eq:cond1} are the required conditions such that the marginal induced densities are uniform and the bivariate density is proper, respectively. Note that \eqref{eq:dens} resembles the sample copula \eqref{eq:samplec}, however, \eqref{eq:dens} is a parametrised model, whereas \eqref{eq:samplec} is a nonparametric estimator of a bivariate copula. 

Conditions \eqref{eq:cond1} constrain the parameter space $\Theta$ leaving us with a reduced number of parameters. That is, instead of having $m^2$ parameters, we end up having $(m-1)^2$ free parameters,  $\{\theta_{j,k}\}$ for $j,k=1,\ldots,m-1$, where the boundary parameters are defined as 
\begin{equation}
\label{eq:cond2}
\theta_{j,m}=\frac{1}{m}-\sum_{k=1}^{m-1}\theta_{j,k}\,,\quad\theta_{m,k}=\frac{1}{m}-\sum_{j=1}^{m-1}\theta_{j,k}\quad\mbox{and}\quad
\theta_{m,m}=\sum_{j=1}^{m-1}\sum_{k=1}^{m-1}\theta_{j,k}-\frac{m-2}{m}
\end{equation}
for $j,k=1,\ldots,m-1$. In this case, the parameter space $\Theta$ is defined by the following constraints for the free parameters 
\begin{equation}
\label{eq:cond3}
0<\sum_{j=1}^{m-1}\theta_{j,k}<\frac{1}{m},\;\forall k,\quad 0<\sum_{k=1}^{m-1}\theta_{j,k}<\frac{1}{m},\;\forall j\quad\text{and
}\quad \frac{m-2}{m}<\sum_{j=1}^{m-1}\sum_{k=1}^{m-1}\theta_{j,k}<\frac{m-1}{m}.
\end{equation}

The corresponding copula can be obtained as the CDF of the copula density \eqref{eq:dens}. This has the expression
\begin{equation}
\label{eq:copula}
C(u,v\mid\btheta)=\sum_{j=1}^m\sum_{k=1}^m\left(A_{j,k}+B_{j,k}u+D_{j,k}v+m^2\theta_{j,k}uv\right)I\left((u,v)\in Q_{j,k}\right),
\end{equation}
where 
$$A_{j,k}=\sum_{r=1}^j\sum_{s=1}^k\theta_{r,s}-j\sum_{s=1}^k\theta_{j,s}-k\sum_{r=1}^j\theta_{r,k}+jk\theta_{j,k},$$ $$B_{j,k}=m\sum_{s=1}^k\theta_{j,s}-mk\theta_{j,k}\quad\mbox{and}\quad D_{j,k}=m\sum_{r=1}^j\theta_{r,k}-mj\theta_{j,k}.$$ 

Although the copula density \eqref{eq:dens} is piecewise constant, the corresponding copula \eqref{eq:copula} is absolutely continuous. Moreover, Spearman's $\rho$ coefficient has a simple expression
\begin{equation}
\label{eq:rho}
\rho(\btheta)=\frac{3}{m^2}\left\{4\sum_{j=1}^m\sum_{k=1}^m jk\theta_{j,k}-(m+1)^2\right\}.
\end{equation}  
A similar expression as \eqref{eq:rho} was obtained by \cite{gonzalez&hernandez:13} for the sample copula. 

To establish a connection with the sample copula, we provide the maximum likelihood estimators of the model parameters, which are given in the following Proposition. 
\begin{proposition}
\label{prop:mle}
Let $(U_i,V_i)$, $i=1,\ldots,n$ be a bivariate sample of size $n$ from copula density \eqref{eq:dens}. The maximum likelihood estimators (MLEs) $\widehat{\theta}_{i,j}$ of the parameters $\theta_{j,k}$, for $j,k=1,\ldots,m-1$, satisfy
\begin{equation}
\label{eq:mle}
\frac{r_{j,k}}{\widehat\theta_{j,k}}+\frac{r_{m,m}}{\sum_{t=1}^{m-1}\sum_{s=1}^{m-1}\widehat\theta_{t,s}-(m-2)/m}=\frac{r_{j,m}}{1/m-\sum_{s=1}^{m-1}\widehat\theta_{j,s}}+\frac{r_{m,k}}{1/m-\sum_{t=1}^{m-1}\widehat\theta_{t,k}}
\end{equation}
where $r_{j,k}$, for $j,k=1,\ldots,m$, are given in \eqref{eq:r}. 
\end{proposition}
\begin{proof}
Given the observed sample $\bu=\{u_i\}$ and $\bv=\{v_i\}$, the log-likelihood function for $\btheta$ is given by
$$\hspace{-3.5cm}\log f(\bu,\bv\mid\btheta) = 2n\log(m)+\sum_{j=1}^{m-1}\sum_{k=1}^{m-1} r_{j,k}\log\left(\theta_{j,k}\right)$$
$$\hspace{3cm}+\sum_{j=1}^{m-1}r_{j,m}\log(\theta_{j,m})+ \sum_{k=1}^{m-1}r_{m,k}\log(\theta_{m,k})+ r_{m,m}\log(\theta_{m,m}),$$
where $\theta_{j,m}$, $\theta_{m,k}$ and $\theta_{m,m}$ are given in \eqref{eq:cond2}. We take first derivative with respect to $\theta_{j,k}$ and obtain
$$\frac{\partial}{\partial\theta_{j,k}}\log f=\frac{r_{j,k}}{\theta_{j,k}}-\frac{r_{j,m}}{\theta_{j,m}}-\frac{r_{m,k}}{\theta_{m,k}}+\frac{r_{m,m}}{\theta_{m,m}},$$
for $j,k=1,\ldots,m-1$. After equating the first derivative to zero we obtain condition \eqref{eq:mle}. To prove that the critical point is a maximum we further take second derivative and obtain 
$$\frac{\partial^2}{\partial\theta_{j,k}^2}\log f=-\frac{r_{j,k}}{\theta_{j,k}^2}-\frac{r_{j,m}}{\theta_{j,m}^2}-\frac{r_{m,k}}{\theta_{m,k}^2}-\frac{r_{m,m}}{\theta_{m,m}^2}$$
which is clearly negative. 
\end{proof}

Proposition \ref{prop:mle} provides conditions to obtain the MLEs of parameters $\theta_{j,k}$. However these conditions rely on nonlinear equations. For the specific case of $m=2$, we can obtain explicit analytic expressions. Condition \eqref{eq:mle} simplifies to $$\frac{r_{1,1}}{\widehat\theta_{1,1}}+\frac{r_{2,2}}{\widehat\theta_{1,1}}=\frac{r_{1,2}}{1/2-\widehat\theta_{1,1}}+\frac{r_{2,1}}{1/2-\widehat\theta_{1,1}},$$
which after some algebra we obtain $$\widehat\theta_{1,1}=\frac{r_{1,1}+r_{2,2}}{n}.$$
This is an interesting result since the MLE of the unique parameter of the model (when $m=2$), $\theta_{1,1}$, is an average of the number of points that lie in opposite regions $Q_{1,1}$ and $Q_{2,2}$. 

In practice we do not observe data directly from the copula, that is, $(U_i,V_i)$, $i=1,\ldots,n$ with support in $[0,1]^2$. What we usually observe are data $(X_i,Y_i)$, $i=1,\ldots,n$, with support $\Omega\subset\RR^2$ coming from the CDF $H(x_i,y_i)=C(F(x_i),G(y_i))$, with $F$ and $G$ the marginal CDFs of each coordinate, respectively. However, as in \cite{deheuvels:79}, we can obtain a modified sample $(U_i,V_i)$ using the modified rank transformation \eqref{eq:rank}, to estimate the copula. In this case we have another interesting case from Proposition \ref{prop:mle}, whose result is given in the following corollary. 

\begin{corollary}
Let $(U_i,V_i)$, $i=1,\ldots,\ldots,n$ be a bivariate modified rank transformed sample of size $n$ for data $(X_i,Y_i)$ coming from copula density \eqref{eq:dens} or coming from CDF $H(x_i,y_i)=C(F(x_i),G(y_i))$. Additionally, if $m$ divides $n$, the MLEs of the parameters  $\theta_{j,k}$, reduce to $$\widehat\theta_{j,k}=\frac{r_{j,k}}{n},$$ for $j,k=1,\ldots,m-1$, where $r_{j,k}$ is given in \eqref{eq:r}. Furthermore, the MLE of the copula density $\widehat{f}_C(u,v\mid\btheta)=f_C(u,v\mid\widehat\btheta)$, reduces to the sample copula \eqref{eq:samplec} of \cite{gonzalez&hoyos:18}. 
\end{corollary}
\begin{proof}
For modified rank transformed data and when $m$ divides $n$ the following marginal conditions are satisfied: $\sum_{j=1}^m r_{j,k}=\sum_{k=1}^m r_{j,k}=\frac{n}{m}$. We only need to prove that $\widehat\theta_{j,k}=\frac{r_{j,k}}{n}$ for $j,k=1,\ldots,m-1$ satisfy condition \eqref{eq:mle}. Working on the boundary elements and considering the hypothesis,  $\sum_{t=1}^{m-1}\sum_{s=1}^{m-1}\widehat\theta_{t,s}-(m-2)/m$ becomes $r_{m,m}/n$; $1/m-\sum_{s=1}^{m-1}\widehat\theta_{j,s}$ becomes $r_{j,m}/n$; and $1/m-\sum_{t=1}^{m-1}\widehat\theta_{t,k}$ becomes $r_{m,k}/n$. Substituting these values into \eqref{eq:mle} we obtain that $2n=2n$ which is clearly true. 
\end{proof}

To carry out inference on the model parameters, we suggest to follow a Bayesian approach instead.

\section{Prior distributions}
\label{sec:prior}

Here we propose a prior that recognises dependence across $\theta_{j,k}$'s that belong to neighbouring regions $Q_{j,k}$ in the partition grid. The most common prior for spatial dependence in areas is the conditionally autoregressive model \citep{besag:74}, however this model is defined in terms of normal distributions and the marginal support is the real line. In our case the parameter space for each $\theta_{j,k}$ is a subset of the interval $[0,1]$, so we propose an alternative prior that extends the work of \cite{jara&al:13} and uses ideas from \cite{nieto&bandyopadhyay:13}. 

Let $\partial_{j,k}$ be the set of indexes of spatial neighbours of region $Q_{j,k}$, $j,k=1,\ldots,m$. Since all regions are rectangles, for the purpose of this work, two regions will be considered neighbours if they share an edge, for instance, region defined by indexes $(j,k)$ will have a set of neighbours $\partial_{j,k}=\{(j,k),(j,k-1),(j,k+1),(j-1,k),(j+1,k)\}$. This is illustrated in Figure \ref{fig:part} with grey shadows showing the neighbouring regions of location $(3,3)$. Note that a region is considered a neighbour of itself and that regions $(j,k)$ located at the boundaries of the grid will have less than five neighbours.

Instead of defining the dependence directly on the $\{\theta_{j,k}\}$, we will rely on a set of latent parameters $\{\eta_{j,k}\}$ associated to each of the regions $(j,k)$ in the grid. This latter set will be conditionally independent given a common parameter $\omega$. Therefore, our spatial dependence prior is based on conjugate distributions, in a Bayesian context, and is defined through a three-level hierarchical model of the form
\begin{align}
\nonumber
\theta_{j,k}\mid\bfeta&\stackrel{\mbox{\scriptsize{ind}}}{\sim}\be\left(a+\sum_{(r,s)\in\partial_{j,k}}\eta_{r,s}\,,\:b+\sum_{(r,s)\in\partial_{j,k}}\left(c_{r,s}-\eta_{r,s}\right)\right)\\
\label{eq:prior}
\eta_{j,k}\mid\omega&\stackrel{\mbox{\scriptsize{ind}}}{\sim}\bin\left(c_{j,k}\,,\:\omega\right)\\
\nonumber
\omega&\sim\be(a,b),
\end{align}
where $a,b>0$ and $c_{j,k}\in\Naa$, for $j,k=1,\ldots,m-1$. We will refer to prior \eqref{eq:prior} as spatial beta process and will denote it by $\sbe(a,b,\bc)$, where $\bc=\{c_{j,k}\}$. 

The reason for requiring a three level hierarchical model becomes clear when we study its properties. In Particular, the marginal distribution induced for each $\theta_{j,k}$ and the correlation between any two of them are given in the following proposition. 

\begin{proposition}
\label{prop:cor}
Let $\btheta=\{\theta_{j,k}\}\sim\sbe(a,b,\bc)$ given in \eqref{eq:prior}. Then, $\theta_{j,k}\sim\be(a,b)$ marginally for all $j,k=1,\ldots,m-1$. Moreover, the correlation between any two parameters, say $\theta_{j,k}$ and $\theta_{j',k'}$ is given by 
$$\Cor\left(\theta_{j,k},\theta_{j',k'}\right)=\frac{(a+b)\left(\sum_{(r,s)\in\partial_{j,k}\cap\partial_{j',k'}}c_{r,s}\right)+\left(\sum_{(r,s)\in\partial_{j,k}}c_{r,s}\right)\left(\sum_{(r,s)\in\partial_{j',k'}}c_{r,s}\right)}{\left(a+b+\sum_{(r,s)\in\partial_{j,k}}c_{r,s}\right)\left(a+b+\sum_{(r,s)\in\partial_{j',k'}}c_{r,s}\right)}.$$ 
\end{proposition}
\begin{proof}
To obtain the marginal distribution of $\theta_{j,k}$, we note that the conditional distribution of the sum of the latent variables is $\sum_{(r,s)\in\partial_{j,k}}\eta_{r,s}\mid\omega\sim\bin(\sum_{(r,s)\in\partial_{j,k}}c_{r,s},\omega)$, so unconditionally its distribution is a beta-binomial  $\bebin(a,b,\sum_{(r,s)\in\partial_{j,k}}c_{r,s})$. Therefore, due to conjugacy \citep[e.g.][]{bernardo&smith:00}, it follows that the marginal distribution of $\theta_{j,k}$ is beta. For the correlation we use conditional expectation formula twice to obtain the covariance and use the marginal distribution result to obtain the variances. 
\end{proof}

Another important property of our prior is that if $c_{j,k}=c$, that is, common for all $j,k=1,\ldots,m-1$, the latent parameters $\{\eta_{j,k}\}$ are exchangeable, and the joint distribution of $\{\theta_{j,k}\}$ becomes a strictly stationary process. Additionally, if $c_{j,k}=0$ for all $(j,k)$, then the $\theta_{j,k}$'s become all independent. Therefore for $c_{j,k}>0$, the spatial beta process \eqref{eq:prior} defines a prior for parameters in the bounded support $[0,1]$ and with dependence across neighbouring $\theta_{j,k}$'s according to the set $\partial_{j,k}$. This will produce a smoothing effect in the Bayesian estimation of parameters $\btheta$ of copula model \eqref{eq:dens}.

\section{Posterior distributions}
\label{sec:post}

Let $(U_i,V_i)$, $i=1,\ldots,n$ be a bivariate sample of size $n$ from  copula density \eqref{eq:dens}. Then the likelihood function in terms of the $(m-1)^2$ free parameters has the form
\begin{equation}
\nonumber
f(\bu,\bv\mid\btheta)=m^{2n}\theta_{m,m}^{r_{m,m}}
\left\{\prod_{j=1}^{m-1}\theta_{j,m}^{r_{j,m}}\right\}
\left\{\prod_{k=1}^{m-1}\theta_{m,k}^{r_{m,k}}\right\}
\prod_{j=1}^{m-1}\prod_{k=1}^{m-1}\theta_{j,k}^{r_{j,k}}
\end{equation}
where $\btheta\in\Theta$, as in \eqref{eq:cond3}, and the boundary parameters $\theta_{m,m}$, $\theta_{j,m}$ and $\theta_{m,k}$ are given in \eqref{eq:cond2} and $r_{j,k}$, for $j,k=1,\ldots,m$, are defined in \eqref{eq:r}.

We assume the prior distribution for the $\theta_{j,k}$'s is a spatial beta process $\sbe(a,b,\bc)$, given in \eqref{eq:prior}. Therefore the extended prior distribution, similar to a data augmentation technique \citep{tanner:91}, considering the latent variables $\bfeta=\{\eta_{j,k}\}$ and $\omega$, is given by 
\begin{align}
\nonumber
f(\btheta,\bfeta,\omega)=\prod_{j=1}^{m-1}\prod_{k=1}^{m-1}&\{\be(\theta_{j,k}\mid a+\sum_{(r,s)\in\partial_{j,k}}\eta_{r,s}\,,\:b+\sum_{(r,s)\in\partial_{j,k}}(c_{r,s}-\eta_{r,s}) )\\
\nonumber
&\times\bin(\eta_{j,k}\mid c_{j,k}\,,\omega)\}\,\be(\omega\mid a,b).
\end{align}

The posterior distribution of $(\btheta,\bfeta,\omega)$ is given by the product of the likelihood and the prior, up to a proportionality constant. In order to characterise the posterior distribution, we implement a Gibbs sampler \citep{smith&roberts:93} and sample $(\btheta,\bfeta,\omega)$ from the following conditional posterior distributions.
\begin{enumerate}
\item[(i)] Posterior conditional distribution for $\theta_{j,k}$, $j,k=1,\ldots,m-1$
$$\hspace{-8mm}f(\theta_{j,k}\mid\rest)\propto \theta_{j,k}^{a+\sum_{(r,s)\in\partial_{j,k}}\eta_{r,s}+r_{j,k}-1}\left(1-\theta_{j,k}\right)^{b+\sum_{(r,s)\in\partial_{j,k}}(c_{r,s}-\eta_{r,s})-1}\theta_{m,m}^{r_{m,m}}\theta_{j,m}^{r_{j,m}}\theta_{m,k}^{r_{m,k}}I_{\Theta}(\theta_{j,k}).$$
\item[(ii)] Posterior conditional distribution for $\eta_{j,k}$, $j,k=1,\ldots,m-1$
$$f(\eta_{j,k}\mid\rest)\propto\frac{ {c_{j,k}\choose\eta_{j,k}}\left\{\left(\frac{\omega}{1-\omega}\right)\prod_{(t,s)\in\varrho_{j,k}}\left(\frac{\theta_{t,s}}{1-\theta_{t,s}}\right)\right\}^{\eta_{j,k}}I_{\{0,\ldots,c_{j,k}\}}(\eta_{j,k})}
{\prod_{(t,s)\in\varrho_{j,k}}\Gamma\left(a+\sum_{(l,z)\in\partial_{t,s}}\eta_{l,z}\right)\Gamma\left(b+\sum_{(l,z)\in\partial_{t,s}}(c_{l,z}-\eta_{l,z})\right)},$$
where $\varrho_{j,k}$ is the set of reversed neighbours, that is, the set of pairs $(t,s)$ such that $(j,k)\in\partial_{t,s}$.
\item[(iii)] Posterior conditional distribution for $\omega$
$$f(\omega\mid\rest)=\be\left(\omega\,\left|\,a+\sum_{j=1}^{m-1}\sum_{k=1}^{m-1}\eta_{j,k}\,,\,b+\sum_{j=1}^{m-1}\sum_{k=1}^{m-1}(c_{j,k}-\eta_{j,k})\right.\right).$$
\end{enumerate}
Looking at posterior conditional (i) we realise that the sum of latent variables $\sum\eta_{r,s}$ appears in the posterior in the same way as the data $r_{j,k}$. Moreover, since $\eta_{r,s}\in\{0,\ldots,c_{r,s}\}$ and considering that $\partial_{j,k}$ has between three to five elements, to avoid overwhelming the data, it is advised to take values $c_{j,k}\leq\sqrt{n}/5$. 

Sampling from (iii) is straightforward and sampling from (ii) can be easily done by evaluating at the different points of the support and normalizing. However, sampling from (i) is not trivial and requires a  Metropolis-Hastings step \citep{tierney:94}. We suggest sampling $\theta_{j,k}^*$ at iteration $(t+1)$ from a random walk proposal distribution 
$$q(\theta_{j,k}\mid\btheta_{-(j,k)},\theta_{j,k}^{(t)})=\un\left(\theta_{j,k}\mid \max\{l_{j,k},\theta_{j,k}^{(t)}-\delta_{j,k} d_{j,k}\},\min\{u_{j,k},\theta_{j,k}^{(t)}+\delta_{j,k} d_{j,k}\}\right)$$
where the interval $(l_{j,k},u_{j,k})$ represents the conditional support of $\theta_{j,k}$, $d_{j,k}=u_{j,k}-l_{j,k}$ is its length, with 
$$l_{j,k} = \max \left\{0, \frac{m-2}{m} - \sum_{r=1}^{m-1} \sum_{s=1}^{m-1} \theta_{r,s} I\left((r,s)\neq(j,k)\right) \right\}$$
and
$$u_{j,k} = \min \left\{\frac{m-1}{m} - \sum_{r=1}^{m-1} \sum_{s=1}^{m-1} \theta_{r,s}I\left((r,s)\neq(j,k)\right)\,, \frac{1}{m} - \sum_{s=1,s \ne k}^{m-1}\theta_{j,s}\,, \frac{1}{m} - \sum_{r=1,r \ne j}^{m-1} \theta_{r,k} \right\},$$
for $j,k=1,\ldots,m-1$. Therefore, at iteration $(t+1)$ we accept $\theta_{j,k}^*$ with probability 
$$\alpha\left(\theta_{j,k}^*,\theta_{j,k}^{(t)}\right)=\min\left\{1\,,\;\frac{f(\theta_{j,k}^*\mid\rest)\,q(\theta_{j,k}^{(t)}\mid\btheta_{-(j,k)},\theta_{j,k}^{*})}{f(\theta_{j,k}^{(t)}\mid\rest)\,q(\theta_{j,k}^*\mid\btheta_{-(j,k)},\theta_{j,k}^{(t)})}\right\}.$$

The parameters $\delta_{j,k}$ are tuning parameters that control the acceptance rate. As suggested by \cite{roberts&rosenthal:09}, dropping the subindex, $\delta$ parameter can be adapted every certain amount of iterations, inside the MCMC algorithm, to achieve a target acceptance rate. Differing slightly from the proposal in \cite{roberts&rosenthal:09}, instead of considering a single target acceptance rate, we consider the interval $[0.3, 0.4]$ which, according to \cite{robert&casella:10}, define optimal acceptance rates in random walk MH steps. Specifically, our adaptation method uses batches of 50 iterations and for every batch $b$, we compute the acceptance rate $AR^{(b)}$ and define 
\begin{equation}
\label{eq:adapt}
\delta^{(b+1)}=\left\{\begin{array}{ll}
\min\{1,\delta^{(b)}(1.01)^{\sqrt{b}}\} & \mbox{if } AR^{(b)}>0.4 \\
\max\{0.01,\delta^{(b)}(1.01)^{-\sqrt{b}}\} & \mbox{if } AR^{(b)}<0.3
\end{array}\right. 
\end{equation}
For the examples considered here we use in all cases $\delta^{(1)}=0.25$ as starting value. 

This algorithm was implemented in Python. Figure \ref{fig:batch} shows the performance of this adaptive method for the parameter $\theta_{1,1}$ in the real data analysis, with $m=5$ and $c=2$, of Section \ref{sec:numerical}. The left panel shows the values of tuning parameter $\delta_{1,1}$ that stabilises around $0.9$. The right panel shows that the acceptance rate is kept around the target interval $[0.3,0.4]$ as desired. 

Finally, as a general advise, we suggest to take values for $m$ a lot smaller than $\sqrt{n}$, such that there is at least one data point in each partitioning set. 

\section{Numerical analyses}
\label{sec:numerical}

\subsection{Simulation study}

We first assess the performance of our model in a controlled scenario. For this we consider five families of copulas: Product, Gumbel, Clayton, Ali-Mikhail-Haq (AMH) and Normal. From these families we generated samples of size $n = 200$ with parameters, $\theta=1.3$ for the Gumbel, $\theta \in \{-0.3,1\}$ for the Clayton, $\theta \in \{-0.5, 0.7\}$ for the AMH and $\theta \in \{-0.5, 0.5\}$ for the Normal copula. In all but the first two cases, negative/positive parameters induce negative/positive dependence. We use the Spearman's $\rho$ coefficient to characterise the dependence. Since this measure is not available in closed form for all copulas considered, we computed the theoretical value via numerical integration of expression \eqref{eq:tau-rho}. 

For the prior distributions \eqref{eq:prior} we took $a = 0.1$, $b = 0.1$ and a range of values $c_{jk}\in\{0,1,2\}$ to compare among different degrees of prior dependence. For the partition \eqref{eq:partition} we considered two sizes $m\in\{5,8\}$ in such a way that we can compare with the sample copula. We carry out two analysis, one with the original simulated data as it comes from the model, and another with rank transformed data. We implemented an MCMC with the adaptive scheme as described in Section \ref{sec:post}. Chains were ran for 5,000 iterations with a burn-in of 500 and keeping one of every 2nd iteration to produce posterior estimates. Computational times using an Intel core i7 microprocessor average around 65 minutes.

To assess goodness of fit we computed several statistics. The logarithm of the pseudo marginal likelihood (LPML), originally suggested by \cite{geisser&eddy:79}, to assess the fitting of the model to the data. The supremum norm, defined by $\sup_{(u,v)} |C(u,v)-\widehat{C}(u,v)|$ to assess the discrepancy between our posterior estimate (posterior mean) $\widehat{C}(u,v)$ from the true copula $C(u,v)$. We also computed the Spearman's rho coefficient, as in \eqref{eq:rho}, and compare the 95\% interval estimates with the true value. Additionally, as a graphical aid to see the performance of our model, we compare the posterior estimates (posterior mean) of copula densities with the true ones using heat maps.

In Tables \ref{tab:table1} and \ref{tab:table2}, we show the goodness of fit (GOF) statistics with the sampled data as it comes from the models and after applying rank transformation, this latter are indicated with a super index $r$. We note that the LPML statistics are not comparable between original and rank transformed data, however the supremum norms are comparable. We have included the supremum norm for the frequentist sample copula and added the subindex $F$ to differentiate it from that of our Bayesian model that has a subindex $B$. In all cases we observe that the Spearman's rho coefficient $95\%$ interval estimates $\widehat\rho$ contain the true value $\rho$. 

For the Product copula (Table \ref{tab:table1}, first block) the LPML and supremum norm choose the model with $m=5$ and $c=2$ for both, original and rank transformed data. These cases behave similar to the sample copula according to the supremum norm. 
For the Gumbel copula (Table \ref{tab:table1}, second block), there is no agreement between the LPML and the supremum norm, but in any case they both prefer the model with $m=5$ for the rank transformed data case. Interestingly, as in the product copula case, the sample copula obtains a supremum norm slightly smaller than our best Bayesian model, however the Bayesian model with $m=8$ and ranked transformed data obtains a similar supremum norm for $c=0$. 

For the Clayton copula (Table \ref{tab:table1}, third and fourth blocks), with $\theta=-0.3$ and $\theta=1$, the LPML selects the model with $m=5$ and $c=2$, for original and rank transformed data, and in both cases our Bayesian model is superior than the sample copula. We can notice that for $\theta = 1$ the supremum norm is slightly smaller for $m = 8$ than for $m = 5$. 

For the AMH copula (Table \ref{tab:table2}, first and second blocks), with $\theta=-0.5$ and $\theta=0.7$, there is an slight discrepancy between the LPML and the supremum norm. For original data, the LPML chooses the model with $m=5$ and $c=2$, but the supremum norm chooses that with $m=8$ and $c=0$ or $c=1$. For rank transformed data, the best model is that with $m=5$ and $c=2$, for $\theta=-0.5$. In the case of $\theta=0.7$, the LPML selects the model with $m = 5$ and $c=2$, however the sumpremum norm chooses the model with $m = 8$ and $c=0$ or $c=1$. Comparing with the sample copula, our model performs similarly. 

For the normal copula (Table \ref{tab:table2}, third and fourth blocks), something similar to the AMH copula happens. For both values of $\theta$, the LPML prefers the model with $m=5$ and the supremum norm that with $m=8$. In both cases, our best Bayesian model behaves similarly to the sample copula. 

In Figure \ref{fig:heatmaps} we compare the copula density estimates (Bayesian and frequentist) with the true density using heatmaps. For the families shown, product, AMH and normal (across rows), there are some differences between the Bayesian and frequentist (sample copula) estimates. These differences are due to the prior that smooths the intensities by borrowing information from the neighbouring regions. Moreover, in the five families of copulas studied here, none of the GOF statistics select the prior independence case of $c=0$, which confirms the benefit of the prior dependence in the $\theta_{j,k}$'s.

\subsection{Real data analysis}

In the section we show the performance of our model to estimate the dependence between variables in a real life application where data is not obtained directly from the copula but from some arbitrary unknown distribution. 

In Mexico, the pension system is conformed by ten pension fund managers denominated AFOREs (Spanish acronym for \emph{Administradoras de Fondos para el Retiro}), each of these fund managers work with ten investment funds, based on the group age of the worker. On a monthly basis, the National Commission for the Pension System (CONSAR), publishes statistical information and risk metrics that describe the performance of these pension funds in an open data platform, that can be accedes at \linebreak
\texttt{https://www.consar.gob.mx/gobmx/aplicativo/siset/Enlace.aspx}.

The information provided by CONSAR allows workers to choose the AFORE that can provide them with the best benefits in their retirement. Because of its importance, we consider two of these statistics: Net Return Indicator (IRN), which is an indicator of the average of the short, medium and long-term returns offered by a investment fund, above the cost of a life annuity, minus the applicable commissions, and reflects the past performance obtained by the investments in each fund; and the tracking error (ES), an indicator that shows the average difference between the actual investment path fund and the optimal  glide path. 

In general, it is considered that these two variables, IRN and ES, maintain a positive dependency relationship, that is, a higher return may present a higher error (risk). We use our semiparametric copula model to verify this assumption and quantify the possible degree of dependence. If these two variables were independent or negative dependent, workers might be able to freely chose the AFORE that maximises the IRN without incurring in any risk. Available data consists of $n=100$ observations of variables IRN and ES in December of 2021. As a first step we apply the rank transformation given in \eqref{eq:rank} to the original data. In Figure \ref{fig:realdata} we show scatter plots of the original data (upper left panel), and the rank transformed data (upper right panel). Note that the scale of the data changes, but the main features of dependence are maintained.

To define the prior distribution we took the same definitions indicated in the simulations of the previous section, $a=0.1$, $b=0.1$, $c_{j,k} \in \{0,1,2\}$, except for $m$, this parameter is considered to take values $m\in\{4,5\}$, due to the reduction in the number of observations in the sample with respect to the sample size of the previous section. We consider the same MCMC specifications as those used for the simulation study. Our posterior sampling procedure behaves well with good convergence of the chains and the adaptation reaches the desire target. Computational times using an Intel core i7 microprocessor average around 30 minutes.

In Table \ref{tab:realdata} we report some GOF measures, say the Spearman's rho estimate and the LPML. According to the LPML the values $m=5$ and $c=2$ are preferred. The 95\% credible interval estimate of the Spearman's rho is $(0.003,0.282)$, which confirms that the association is positive. For the reference, the sample Spearman's rho takes the value of $0.177$, however there is no way of knowing if this value is significantly positive. Our model confirms that it is. These estimates suggest that there is a positive (weak) dependence between the return and the risk in an investment fund. Therefore, workers must pay attention at the IRN indicator as well as the ES in order to choose a fund manager. 
Finally, in Figure \ref{fig:realdata}, we also show the Bayesian estimators for the copula density as a heatmap (bottom left panel) and for the copula CDF as a perspective plot (bottom right panel). In the heatmap we can appreciate slightly more intense colors in the $45$ degrees diagonal, which confirms the existence of a positive dependence.

\section{Concluding remarks}

We proposed a semiparametric copula model that is flexible enough to approximate the dependence between any two random variables. Maximum likelihood estimators of our model coincide with the sample copula of \cite{gonzalez&hoyos:18} under certain conditions such as rank transformation of the data and defining an $m$ that divides $n$. However, our model is more general and due to the Bayesian analysis, we can produce better estimates by borrowing strength among grid neighbours through the prior distribution. 

Computational times reported are not dependent on the sample size $n$, but they are related to the partition size $m$. The number of parameters and latent parameters to sample from, in the MCMC algorithm, is $2(m-1)^2+1$. We advise to take $m<\min\{\sqrt{n},10\}$ to obtain results with a reasonable amount of time. 

Along this paper we concentrated in the bivariate copula, however the extension to a $d$-dimensional copula can also be considered. For instance, if we consider a partition of size $m^d$, $2\leq m\leq n$, of $[0,1]^d$ such that $Q_{j_1,\ldots,j_d}=\times_{k=1}^d \left(\frac{j_k-1}{m},\frac{j_k}{m}\right]$ for $j_k=1,\ldots,m$ and $k=1,\ldots,d$, then a semiparametric $d$-copula density would be $$f_C(u_1,\ldots,u_d\mid\btheta)=m^d\sum_{j_1=1}^m\cdots\sum_{j_d=1}^m \theta_{j_1,\ldots,j_d}I((u_1,\ldots,u_d)\in Q_{j_1,\ldots,j_d}),$$ where $\btheta=\{\theta_{j_1,\ldots,j_d},j_1,\ldots,j_d=1,\ldots,m\}$ are the set of model parameters that satisfy $\sum_{j_k=1}^m\theta_{j_1,\ldots,j_d}=1$ for all $k=1,\ldots,d$ and $\sum_{j_1}^m\cdots\sum_{j_k=1}^m\theta_{j_1,\ldots,j_d}=1$. Extending the prior to this $d$-dimensional setting is also possible. Performance of our semiparametric copula model in this multivariate setting is worth studying.

\section*{Acknowledgement}

The first author acknowledges support from \textit{Asociaci\'on Mexicana de Cultura, A.C.}

\bibliographystyle{natbib}

\newpage

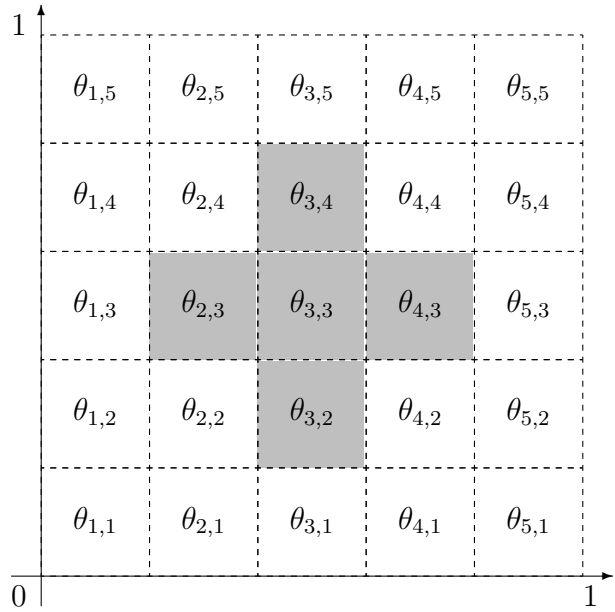
\begin{figure}[H]
\setlength{\unitlength}{0.8cm}
\begin{center}
\begin{picture}(10,10)
\put(0,0.5){\vector(1,0){10}}
\put(0.5,0){\vector(0,1){10}}
\put(0,0){$0$}
\put(9.5,0){$1$}
\put(0,9.5){$1$}
\put(4.1,4.24){\colorbox{lightgray}{\makebox(1.5,1.5){}}}
\put(4.1,6.04){\colorbox{lightgray}{\makebox(1.5,1.5){}}}
\put(5.9,4.24){\colorbox{lightgray}{\makebox(1.5,1.5){}}}
\put(4.1,2.44){\colorbox{lightgray}{\makebox(1.5,1.5){}}}
\put(2.3,4.24){\colorbox{lightgray}{\makebox(1.5,1.5){}}}
\put(0.5,0.5){\dashbox{0.1}(1.8,1.8){$\theta_{1,1}$}}
\put(2.3,0.5){\dashbox{0.1}(1.8,1.8){$\theta_{2,1}$}}
\put(4.1,0.5){\dashbox{0.1}(1.8,1.8){$\theta_{3,1}$}}
\put(5.9,0.5){\dashbox{0.1}(1.8,1.8){$\theta_{4,1}$}}
\put(7.7,0.5){\dashbox{0.1}(1.8,1.8){$\theta_{5,1}$}}
\put(0.5,2.3){\dashbox{0.1}(1.8,1.8){$\theta_{1,2}$}}
\put(2.3,2.3){\dashbox{0.1}(1.8,1.8){$\theta_{2,2}$}}
\put(4.1,2.3){\dashbox{0.1}(1.8,1.8){$\theta_{3,2}$}}
\put(5.9,2.3){\dashbox{0.1}(1.8,1.8){$\theta_{4,2}$}}
\put(7.7,2.3){\dashbox{0.1}(1.8,1.8){$\theta_{5,2}$}}
\put(0.5,4.1){\dashbox{0.1}(1.8,1.8){$\theta_{1,3}$}}
\put(2.3,4.1){\dashbox{0.1}(1.8,1.8){$\theta_{2,3}$}}
\put(4.1,4.1){\dashbox{0.1}(1.8,1.8){$\theta_{3,3}$}}
\put(5.9,4.1){\dashbox{0.1}(1.8,1.8){$\theta_{4,3}$}}
\put(7.7,4.1){\dashbox{0.1}(1.8,1.8){$\theta_{5,3}$}}
\put(0.5,5.9){\dashbox{0.1}(1.8,1.8){$\theta_{1,4}$}}
\put(2.3,5.9){\dashbox{0.1}(1.8,1.8){$\theta_{2,4}$}}
\put(4.1,5.9){\dashbox{0.1}(1.8,1.8){$\theta_{3,4}$}}
\put(5.9,5.9){\dashbox{0.1}(1.8,1.8){$\theta_{4,4}$}}
\put(7.7,5.9){\dashbox{0.1}(1.8,1.8){$\theta_{5,4}$}}
\put(0.5,7.7){\dashbox{0.1}(1.8,1.8){$\theta_{1,5}$}}
\put(2.3,7.7){\dashbox{0.1}(1.8,1.8){$\theta_{2,5}$}}
\put(4.1,7.7){\dashbox{0.1}(1.8,1.8){$\theta_{3,5}$}}
\put(5.9,7.7){\dashbox{0.1}(1.8,1.8){$\theta_{4,5}$}}
\put(7.7,7.7){\dashbox{0.1}(1.8,1.8){$\theta_{5,5}$}}
\end{picture}
\end{center}
\caption{Graphical representation of unit square partition with $m=5$. Neighbouring regions of location $(3,3)$ are painted in gey.}
\label{fig:part}
\end{figure}

\begin{figure}[H]
\begin{center}
\includegraphics[width=6.5cm,height=5.5cm]{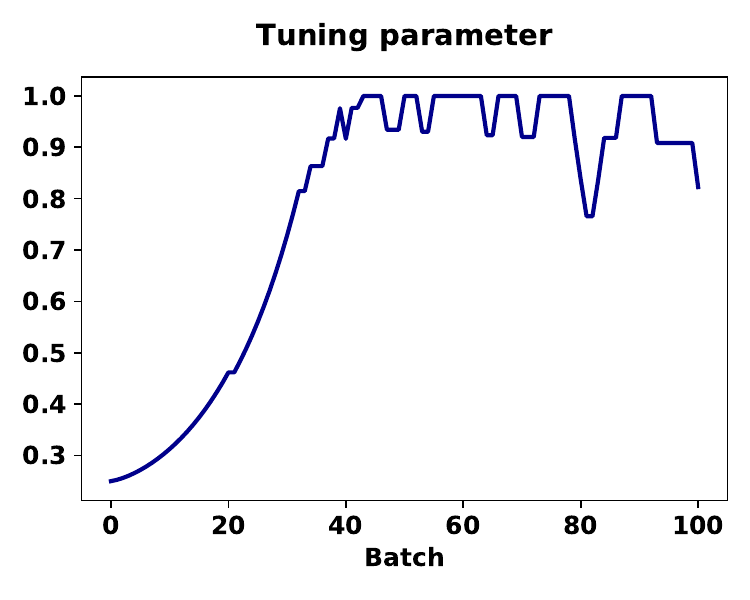} \hspace{1cm}
\includegraphics[width=6.5cm,height=5.5cm]{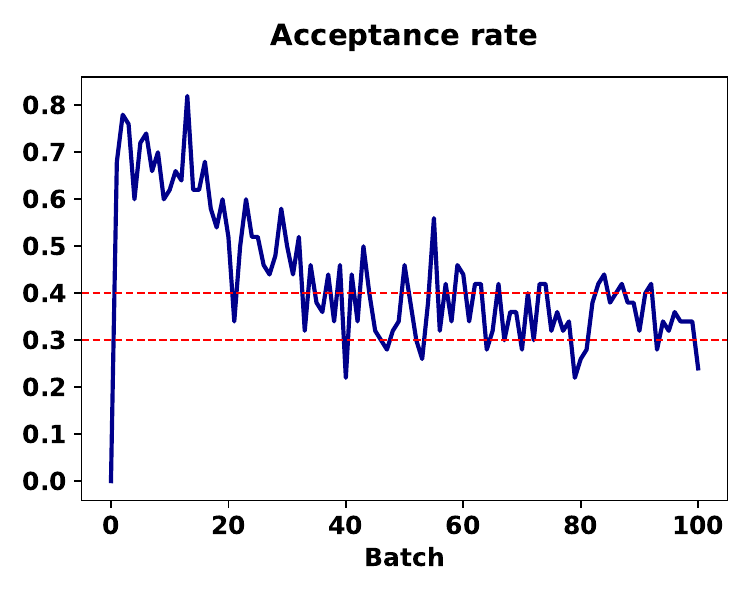}
\caption{Example of acceptance rate and tuning parameter in the adaptation method with a batch every 50 iterations.}
\label{fig:batch}
\end{center}
\end{figure}

\begin{figure}[H]
\begin{center}
\includegraphics[width=5.5cm,height=4.5cm]{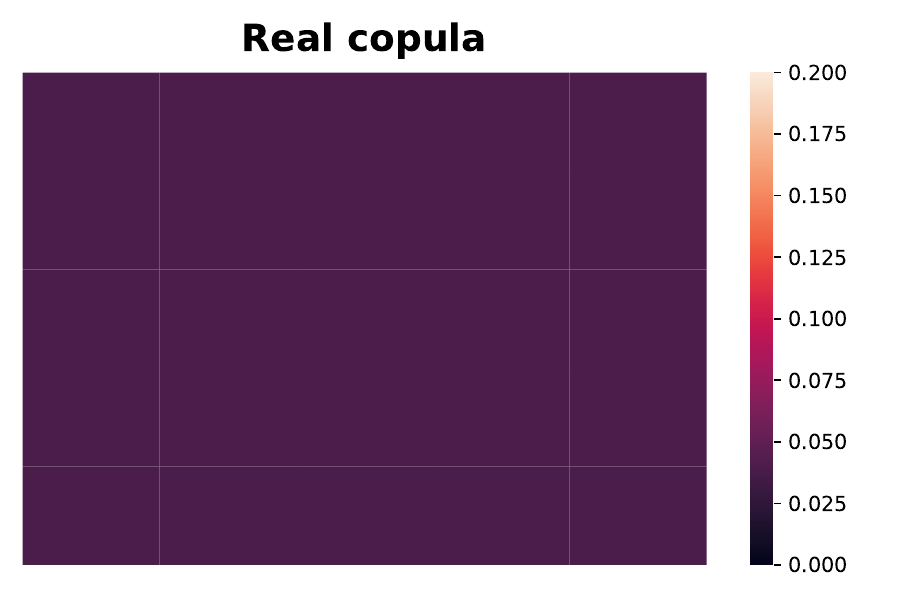}\hspace{-3mm}
\includegraphics[width=5.5cm,height=4.5cm]{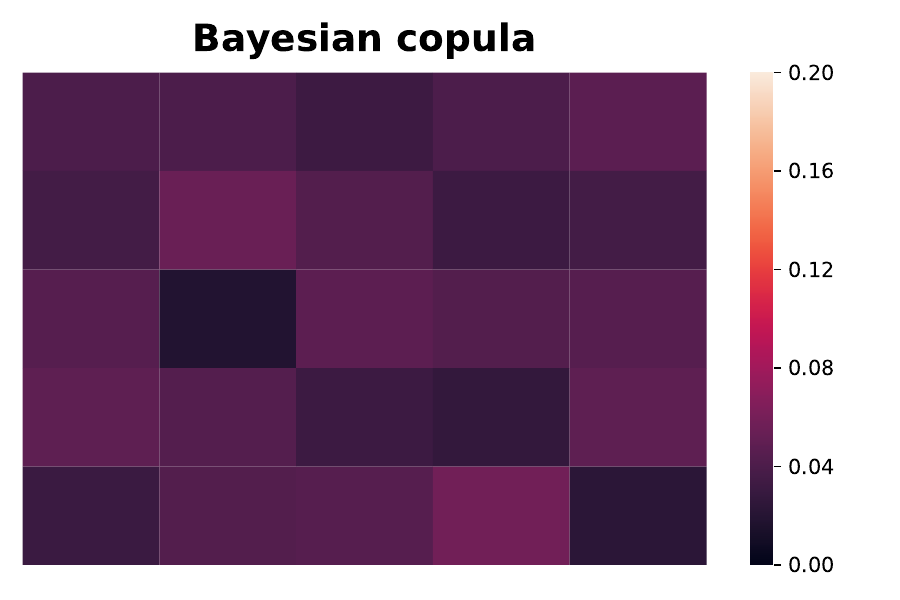}\hspace{-3mm}
\includegraphics[width=5.5cm,height=4.5cm]{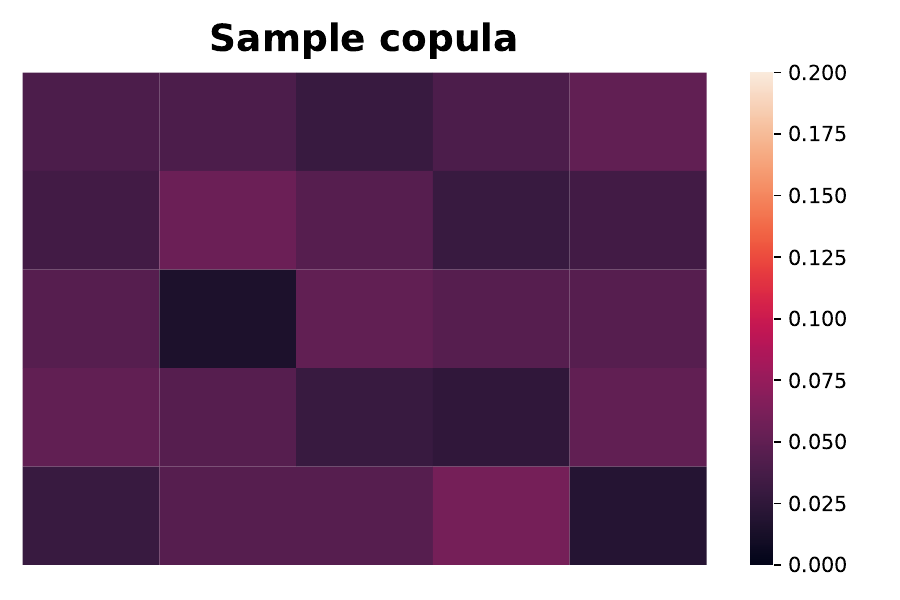}
\includegraphics[width=5.5cm,height=4.5cm]{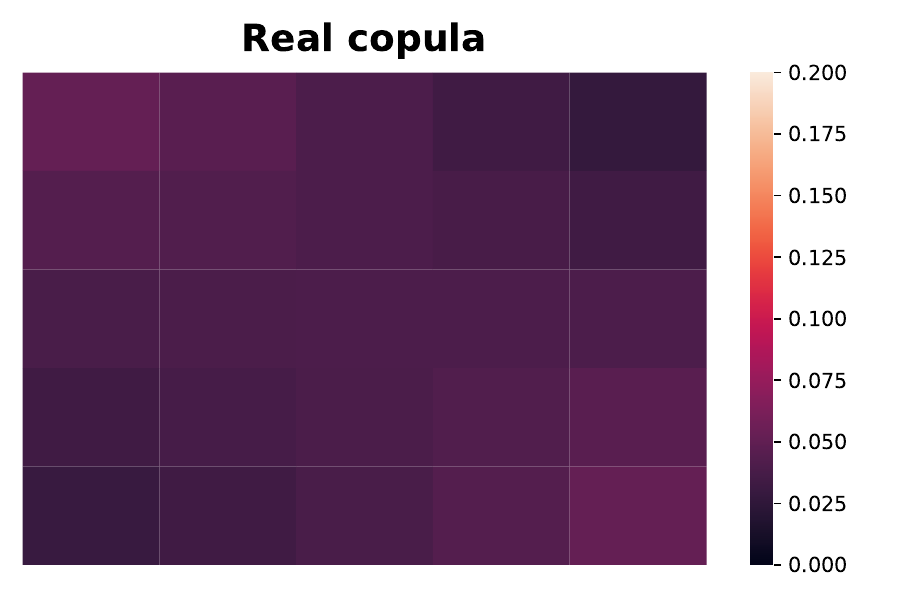}\hspace{-3mm}
\includegraphics[width=5.5cm,height=4.5cm]{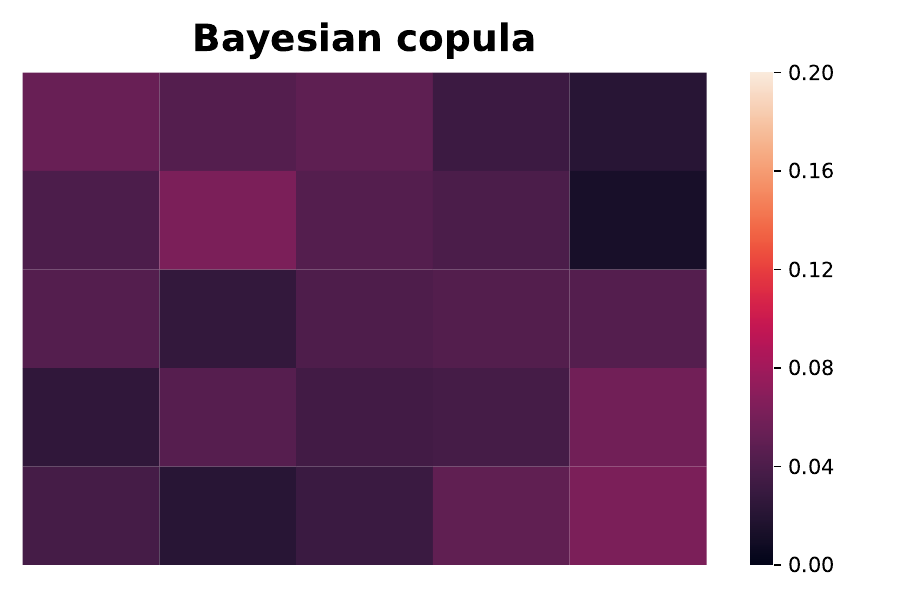}\hspace{-3mm}
\includegraphics[width=5.5cm,height=4.5cm]{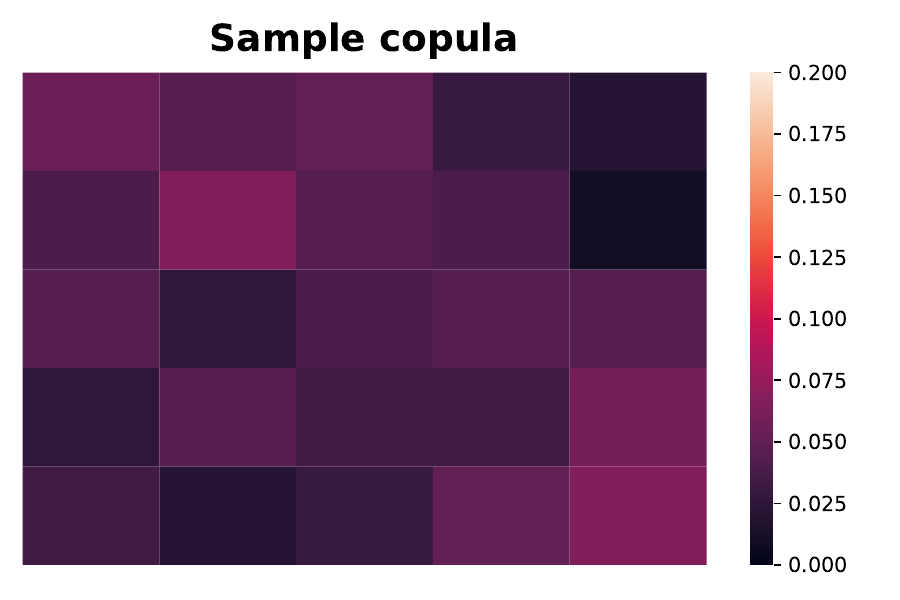}
\includegraphics[width=5.5cm,height=4.5cm]{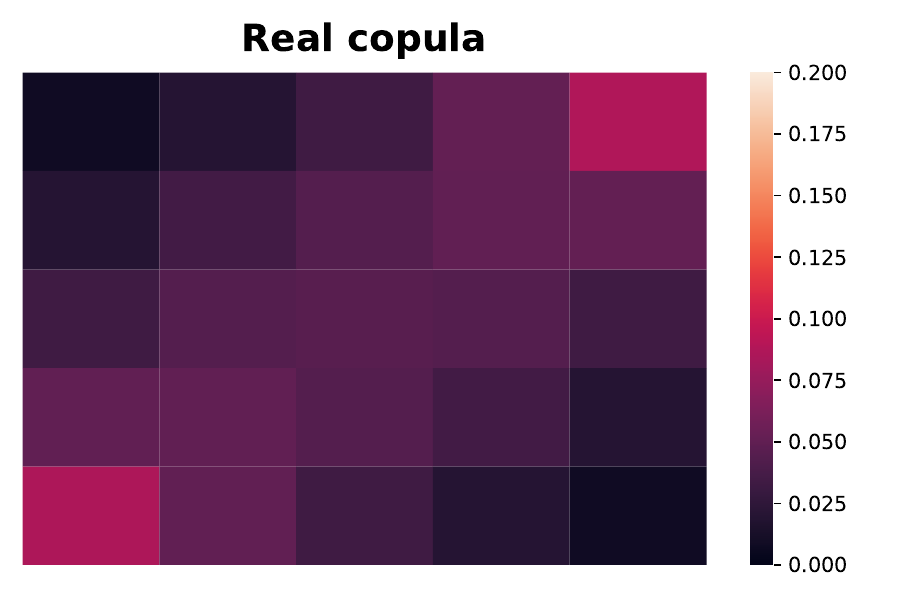}\hspace{-3mm}
\includegraphics[width=5.5cm,height=4.5cm]{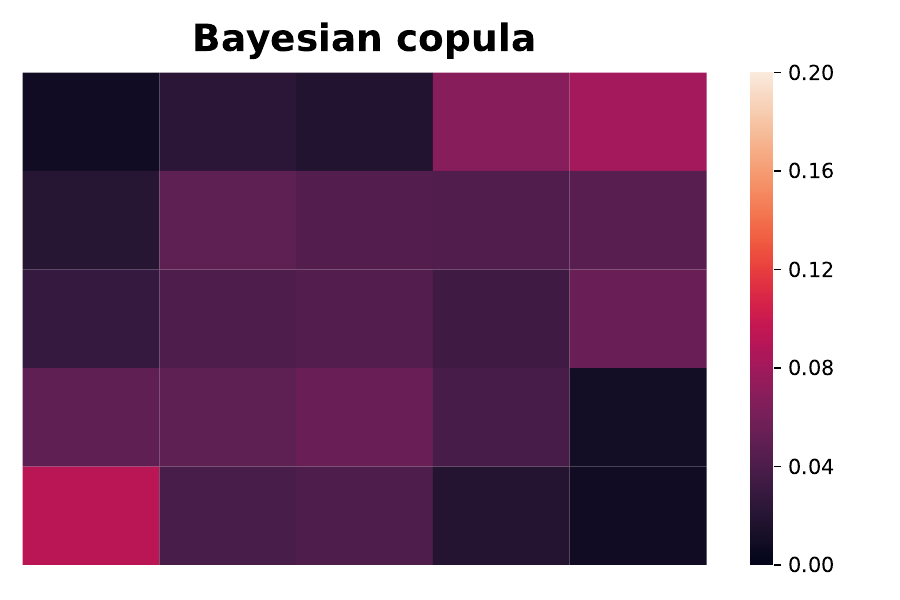}\hspace{-3mm}
\includegraphics[width=5.5cm,height=4.5cm]{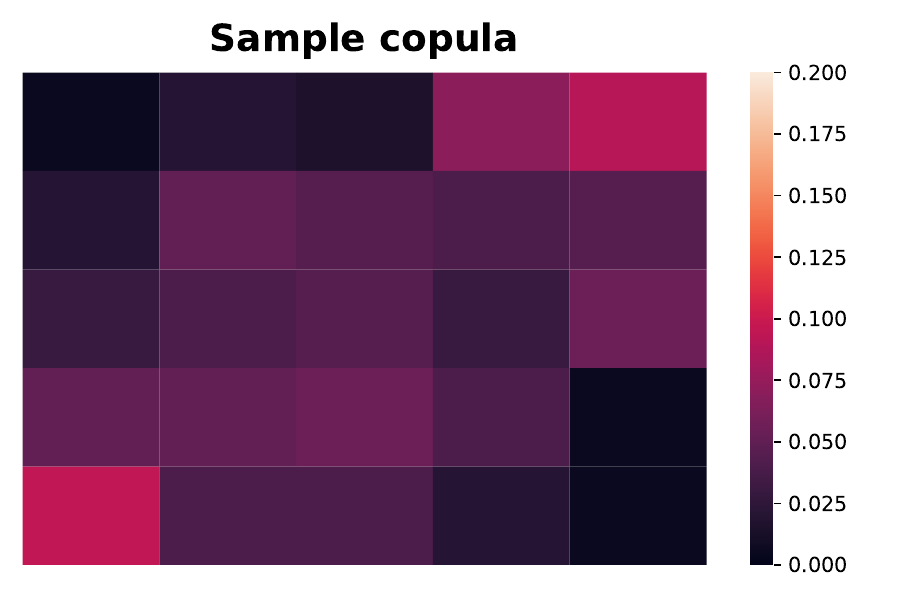}
\caption{Copula density estimation (heat maps) for some copulas based on rank transformed data. Product (top row), AMH with $\theta=-0.5$ (middle row), and normal with $\theta=0.5$ (bottom row). Real copula (first column), bayesian estimation with $m = 5$ and $c = 2$ (second column) and frequentist sample copula (third column).}
\label{fig:heatmaps}
\end{center}
\end{figure}

\begin{figure}
\begin{center}
\includegraphics[width=5.5cm,height=4.5cm]{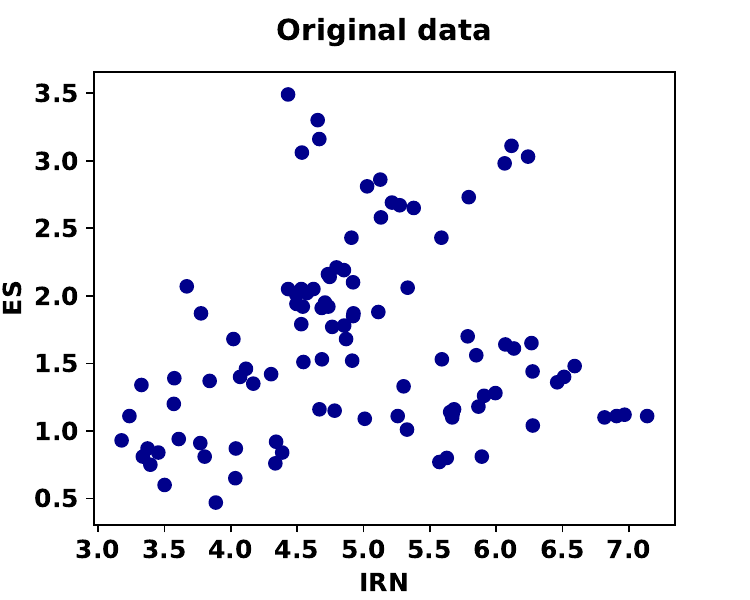} \hspace{1cm}
\includegraphics[width=5.5cm,height=4.5cm]{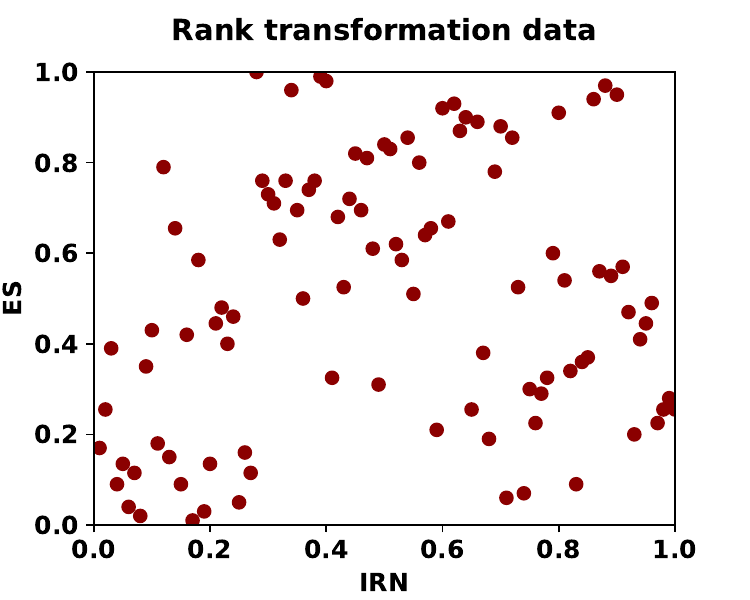} \\[5mm]
\includegraphics[width=5.5cm,height=4.5cm]{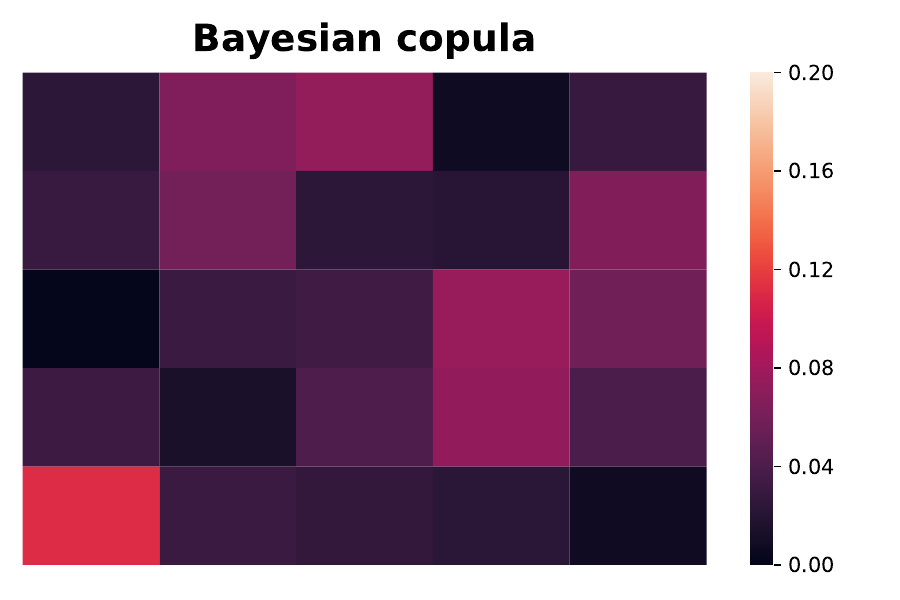} \hspace{0.5cm}
\includegraphics[width=5.5cm,height=4.5cm]{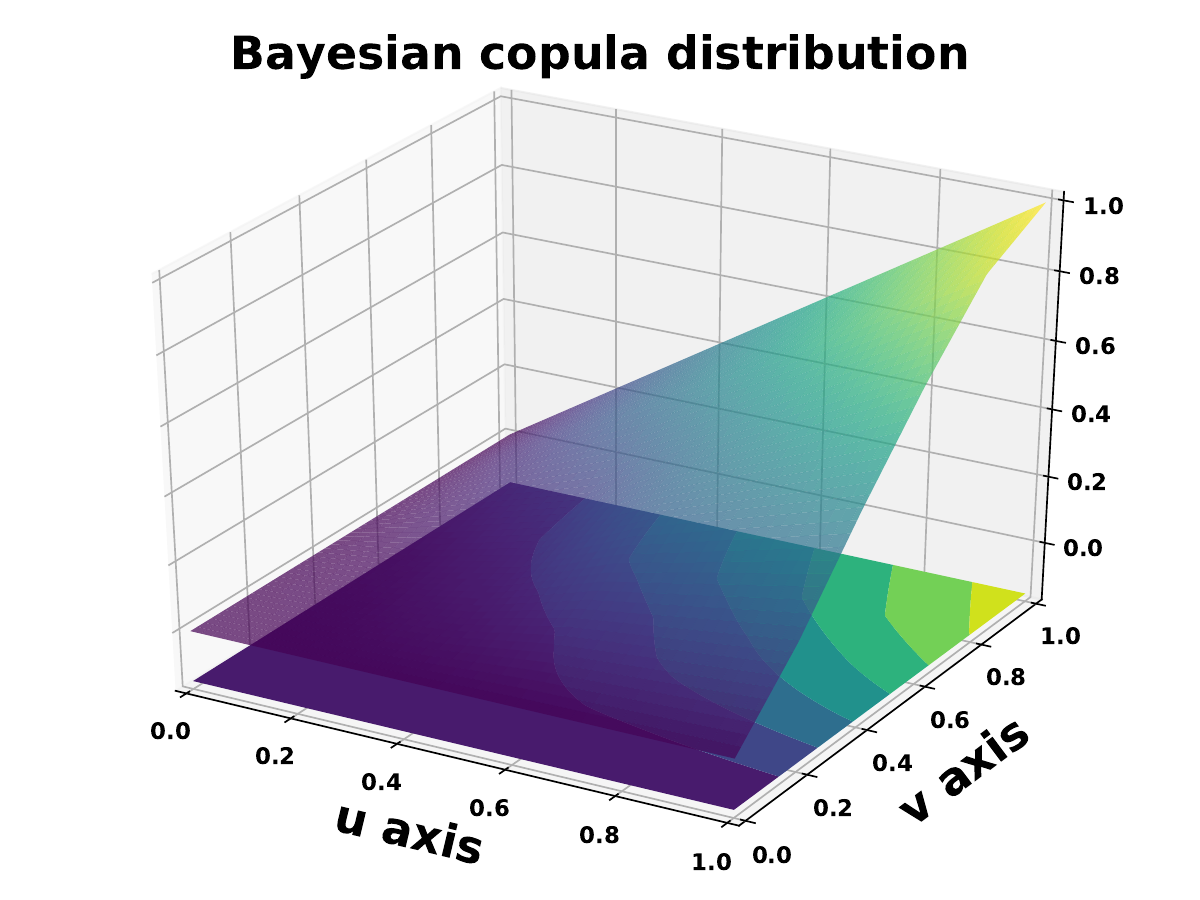}
\caption{Real data. Scatter plots (top row) of original data (left) and rank transformed data (right); Bayesian copula estimators (bottom), density (left) and CDF (right) obtained with $m=5$ and $c=2$.}
\label{fig:realdata}
\end{center}
\end{figure}

\begin{table}
\caption{GOF measures for original and rank transformed data, for different copulas.}
\label{tab:table1}
{\small
\begin{center}
\begin{tabular}{cc|ccc|cccc}
\hline \hline
$m$ & $c$ & $\widehat\rho$ & LPML & SN$_B$ & $\widehat\rho^r$ & LPML$^r$ & SN$_B^r$ & SN$_F^r$\\ \hline
\multicolumn{9}{c}{Product copula, $\rho=0$} \\ \hline
5 & 0 & (-0.140,0.106) & -0.030 & 0.042 & (-0.118,0.128) & -0.039 & 0.038 & 0.029 \\
5 & 1 & (-0.139,0.100) & -0.029 & 0.041 & (-0.105,0.127) & -0.037 & 0.037 & 0.029 \\ 
5 & 2 & (-0.131,0.107) & -0.027 & 0.040 & (-0.105,0.127) & -0.036 & 0.037 & 0.029 \\
8 & 0 & (-0.099,0.147) & -0.134 & 0.044 & (-0.137,0.109) & -0.151 & 0.043 & 0.037 \\
8 & 1 & (-0.097,0.145) & -0.145 & 0.045 & (-0.137,0.118) & -0.173 & 0.043 & 0.037 \\ 
8 & 2 & (-0.097,0.142) & -0.136 & 0.046 & (-0.137,0.118) & -0.159 & 0.042 & 0.042 \\ 
\hline
\multicolumn{9}{c}{Gumbel copula with $\theta=1.3$, $\rho=0.33$} \\ \hline 
5 & 0 & (0.199,0.415) & 0.044 & 0.066 & (0.200,0.430) & 0.044 & 0.067 & 0.058 \\
5 & 1 & (0.195,0.420) & 0.045 & 0.066 & (0.212,0.435) & 0.044 & 0.066 & 0.058 \\ 
5 & 2 & (0.199,0.423) & 0.048 & 0.065 & (0.206,0.421) & 0.050 & 0.066 & 0.058 \\ 
8 & 0 & (0.191,0.416) & -0.053 & 0.059 & (0.204,0.419) & -0.056 & 0.059 & 0.059 \\
8 & 1 & (0.195,0.407) & -0.050 & 0.059 & (0.196,0.416) & -0.040 & 0.057 & 0.059 \\ 
8 & 2 & (0.203,0.419) & -0.042 & 0.057 & (0.221,0.426) & -0.042 & 0.058 & 0.059 \\ 
\hline
\multicolumn{9}{c}{Clayton copula with $\theta=-0.3$, $\rho=-0.26$} \\ \hline 
5 & 0 & (-0.297,-0.064) & -0.025 & 0.064 & (-0.334,-0.117) & 0.010 & 0.085 & 0.090 \\ 
5 & 1 & (-0.300,-0.066) & -0.026 & 0.064 & (-0.329,-0.110) & 0.008 & 0.085 & 0.090 \\ 
5 & 2 & (-0.296,-0.061) & -0.017 & 0.063 & (-0.324,-0.117) & 0.014 & 0.084 & 0.090 \\
8 & 0 & (-0.319,-0.105) & -0.051 & 0.081 & (-0.340,-0.119) & -0.117 & 0.077 & 0.078 \\  
8 & 1 & (-0.325,-0.112) & -0.051 & 0.091 & (-0.326,-0.125) & -0.113 & 0.079 & 0.078\\ 
8 & 2 & (-0.329,-0.110) & -0.045 & 0.084 & (-0.328,-0.120) & -0.109 & 0.077 & 0.078 \\
\hline
\multicolumn{9}{c}{Clayton copula with $\theta=1$, $\rho=0.47$} \\ \hline 
5 & 0 & (0.266,0.480) & 0.063 & 0.088 & (0.272,0.485) & 0.052 & 0.091 & 0.094 \\
5 & 1 & (0.279,0.488) & 0.063 & 0.090 & (0.283,0.489) & 0.052 & 0.093 & 0.094 \\ 
5 & 2 & (0.274,0.480) & 0.066 & 0.087 & (0.273,0.487) & 0.054 & 0.089 & 0.094 \\ 
8 & 0 & (0.286,0.502) & -0.009 & 0.067 & (0.281,0.484) & -0.018 & 0.078 & 0.068 \\
8 & 1 & (0.302,0.502) & -0.009 & 0.067 & (0.278,0.480) & -0.018 & 0.078 & 0.068 \\ 
8 & 2 & (0.301,0.509) & -0.009 & 0.067 & (0.279,0.487) & -0.022 & 0.075 & 0.068 \\ 
\hline \hline
\end{tabular}
\end{center}}
\end{table}

\begin{table}
\caption{GOF measures for original and rank transformed data, for different copulas.}
\label{tab:table2}
{\small
\begin{center}
\begin{tabular}{cc|ccc|cccc}
\hline \hline
$m$ & $c$ & $\widehat\rho$ & LPML & SN$_B$ & $\widehat\rho^r$ & LPML$^r$ & SN$_B^r$ & SN$_F^r$\\ \hline
\multicolumn{9}{c}{AMH copula with $\theta=-0.5$, $\rho=-0.15$} \\ \hline
5 & 0 & (-0.368,-0.125) & 0.006 & 0.071 & (-0.347,-0.088) & -0.023 & 0.057 & 0.048 \\ 
5 & 1 & (-0.370,-0.133) & 0.004 & 0.071 & (-0.355,-0.102) & -0.019 & 0.057 & 0.048 \\ 
5 & 2 & (-0.367,-0.137) & 0.009 & 0.069 & (-0.336,-0.097) & -0.015 & 0.056 & 0.048 \\ 
8 & 0 & (-0.343,-0.130) & -0.076 & 0.048 & (-0.315,-0.094) & -0.058 & 0.049 & 0.050 \\ 
8 & 1 & (-0.344,-0.129) & -0.077 & 0.048 & (-0.324,-0.081) & -0.060 & 0.049 & 0.050 \\ 
8 & 2 & (-0.364,-0.123) & 0.070 & 0.050 & (-0.314,-0.087) & -0.051 & 0.047 & 0.050 \\
\hline
\multicolumn{9}{c}{AMH copula with $\theta=0.7$, $\rho=0.28$} \\ \hline 
5 & 0 & (0.190,0.432) & 0.023 & 0.070 & (0.200,0.416) & 0.032 & 0.074 & 0.064 \\ 
5 & 1 & (0.213,0.442) & 0.024 & 0.071 & (0.204,0.432) & 0.032 & 0.073 & 0.064 \\ 
5 & 2 & (0.209,0.439) & 0.029 & 0.069 & (0.204,0.424) & 0.039 & 0.071 & 0.064 \\ 
8 & 0 & (0.174,0.402) & -0.078 & 0.062 & (0.205,0.417) & -0.076 & 0.056 & 0.047 \\ 
8 & 1 & (0.175,0.404) & -0.080 & 0.062 & (0.199,0.415) & -0.082 & 0.056 & 0.047 \\ 
8 & 2 & (0.190,0.414) & -0.074 & 0.062 & (0.214,0.424) & -0.087 & 0.057 & 0.047 \\ 
\hline
\multicolumn{9}{c}{Normal copula with $\theta=-0.5$, $\rho=-0.48$} \\ \hline 
5 & 0 & (-0.517,-0.330) & 0.106 &  0.109 & (-0.541,-0.352) & 0.095 & 0.101 & 0.095 \\ 
5 & 1 & (-0.506,-0.324) & 0.103 & -0.111 & (-0.538,-0.343) & 0.093 & 0.101 & 0.095 \\ 
5 & 2 & (-0.517,-0.324) & 0.106 & 0.109 & (-0.530,-0.349) & 0.100 & 0.100 & 0.095 \\ 
8 & 0 & (-0.507,-0.333) & -0.044 & 0.097 & (-0.525,-0.336) & 0.005 & 0.092 & 0.085 \\ 
8 & 1 & (-0.491,-0.303) & -0.045 & 0.092 & (-0.539,-0.332) & -0.001 & 0.093 & 0.085 \\ 
8 & 2 & (-0.507,-0.058) & -0.039 & 0.096 & (-0.521,-0.332) & 0.014 & 0.090 & 0.085 \\ 
\hline
\multicolumn{9}{c}{Normal copula with $\theta=0.5$, $\rho=0.48$} \\ \hline 
5 & 0 & (0.356,0.554) & 0.105 & 0.085 & (0.341,0.527) & 0.092 & 0.088 & 0.088 \\ 
5 & 1 & (0.361,0.557) & 0.105 & 0.086 & (0.332,0.539) & 0.090 & 0.088 & 0.088 \\ 
5 & 2 & (0.359,0.557) & 0.108 & 0.085 & (0.323,0.524) & 0.093 & 0.085 & 0.088 \\
8 & 0 & (0.348,0.550) & 0.075 & 0.084 & (0.359,0.544) & 0.030 & 0.082 & 0.076 \\  
8 & 1 & (0.375,0.557) & 0.078 & 0.082 & (0.361,0.535) & 0.023 & 0.086 & 0.076 \\ 
8 & 2 & (0.353,0.562) & 0.074 & 0.083 & (0.355,0.542) & 0.037 & 0.083 & 0.076 \\ 
\hline \hline
\end{tabular}
\end{center}}
\end{table}

\begin{table}
\caption{Real data: GOF measures for rank transformed data.}
\label{tab:realdata}
\begin{center}
\begin{tabular}{cc|cc}
\hline\hline
$m$ & $c$ & $\widehat\rho$ & LPML \\ \hline
4 & 0 & (0.038,0.313) & 0.098 \\ 
4 & 1 & (0.028,0.307) & 0.105 \\ 
4 & 2 & (0.019,0.302) & 0.107  \\ 
\hline
5 & 0 & (0.018,0.275) & 0.125  \\ 
5 & 1 & (0.001,0.274) & 0.118  \\ 
5 & 2 & (0.003,0.282) & 0.132  \\ 
\hline\hline
\end{tabular}
\end{center}
\end{table}

\end{document}